\documentclass[11pt]{article}

\usepackage[a4paper,margin=1in]{geometry}
\usepackage{amsmath,amssymb,amsthm,array,booktabs,microtype}
\usepackage[hidelinks]{hyperref}
\hypersetup{
  pdftitle={Secret Sharing at the Shannon Ceiling},
  pdfauthor={Christopher Williamson},
  pdfsubject={Lower bounds for perfect secret-sharing schemes},
  pdfkeywords={secret sharing, access structures, entropy method, lower bounds}
}

\newtheorem{theorem}{Theorem}
\newtheorem{lemma}[theorem]{Lemma}
\newtheorem{corollary}[theorem]{Corollary}

\newcommand{\Zk}{\mathbb Z_k}
\newcommand{\cP}{\mathcal P}
\newcommand{\cG}{\mathcal G}
\newcommand{\cH}{\mathcal H}
\newcommand{\cC}{\mathcal C}
\newcommand{\cB}{\mathcal B}
\newcommand{\Av}{\operatorname{Av}}

\title{Secret Sharing at the Shannon Ceiling}
\author{
  Christopher Williamson\thanks{\texttt{cw@williamsonchris.com}}
}
\date{18 August 2026}

\begin{document}
\maketitle
\begin{abstract}
For every \(n\geq9\) that is a multiple of \(3\), we construct an explicit access
structure on \(n\) participants.  In every perfect secret-sharing scheme
realising this access structure, if \(S\) denotes the random secret, then the
sum of the share entropies is at least
\[
  \left(\frac{n^2}{9}+\frac{2n}{3}\right)H(S),
\]
and some participant has share entropy at least
\[
  \left(\frac{n}{6}+\frac12\right)H(S).
\]
After normalisation by \(H(S)\), these are respectively
\(\Omega(n^2)\) and \(\Omega(n)\) lower bounds and also give the same
asymptotic lower bounds on the total and largest expected binary lengths of the shares.  This improves by a logarithmic factor the longstanding general lower bounds of
\(\Omega(n^2/\log n)\) for total share size and
\(\Omega(n/\log n)\) for maximum share size due to Csirmaz.

The proof uses only elementary Shannon inequalities, together with some averaging
arguments.  The Shannon-information
method has universal \(O(n^2)\) and \(O(n)\) ceilings for the total and
maximum normalised entropy lower bounds it can certify, so our construction
reaches both ceilings up to constant factors.
\end{abstract}


\section{Introduction}

A perfect secret-sharing scheme distributes a secret among a set of
participants so that every authorised coalition can reconstruct the secret,
whereas every unauthorised coalition has no information about it.  A central
problem is to determine how large the shares must be for the most difficult
access structures.

Csirmaz constructed explicit $n$-participant access structures for which one
share has size $\Omega(n/\log n)$ times the secret size
\cite{Csirmaz1997}.  His related total-share construction gives a lower bound
of $\Omega(n^2/\log n)$ \cite{Csirmaz1996}; this was still described as the
best general total-share lower bound in 2023 \cite{Beimel2023}.  We exhibit
an explicit family attaining linear maximum-share complexity and quadratic
total-share complexity.  To the best of our knowledge, these improve the
previously published general lower bounds by a logarithmic factor.

Applebaum and Nir observed that an explicit monotone formula of size \(L\)
on \(O(\sqrt L)\) variables with total secret-sharing complexity
\(\Omega(L)\) would remove a logarithmic loss from their hardness result
\cite[Remark~2.2]{ApplebaumNir2025}.  Our family has precisely these
parameters.

The access structure has three cyclically indexed classes of participants,
denoted $X_i,Y_i,h_i$.  Each helper $h_s$ selects a perfect matching between
the $X$- and $Y$-participants, and a coalition containing $h_s$ is authorised
when it contains an endpoint of every edge in that matching.  

The argument uses only the standard Shannon inequalities together with
perfect privacy and reconstruction.  It therefore applies to arbitrary
perfect schemes, without linearity or representability assumptions.

\section{Entropy identities and secret-sharing inequalities}

Let $\cP$ be a finite participant set.  A secret-sharing scheme consists of a
jointly distributed secret $S$ and shares $(Z_a)_{a\in\cP}$, with $H(S)>0$.
For a coalition $A\subseteq\cP$, write $Z_A=(Z_a)_{a\in A}$.  To reduce
notation, after this section we write $H(A)$ for $H(Z_A)$ and write a
participant label $a$ in place of its share $Z_a$.  Thus $Aa$ denotes
$A\cup\{a\}$, and juxtaposition denotes union of named coalitions. For random variables \(U_1,\ldots,U_m\), juxtaposition denotes their joint distribution; thus \(H(UW)=H(U,W)\), and more generally \(H(AUW)\) denotes the joint entropy of all random variables represented by \(A\), \(U\), and \(W\).

The scheme is \emph{perfect} for an access structure $\Gamma$ if
\[
 H(S\mid Z_A)=0 \quad(A\in\Gamma),
 \qquad
 I(S;Z_A)=0 \quad(A\notin\Gamma).
\]
For an unauthorised coalition, the second condition is equivalently
$H(S\mid Z_A)=H(S)$.

We use the following standard identities.  For random variables
$U,V,W$:
\begin{align}
 H(U\mid W)&=H(UW)-H(W),                                      \label{eq:cond-H}\\
 I(U;V\mid W)&=H(U\mid W)-H(U\mid VW)                         \label{eq:cmi-one}\\
 &=H(UW)+H(VW)-H(W)-H(UVW).                                    \label{eq:cmi-four}
\end{align}
The last line follows from the preceding two by expanding
\[
 H(U\mid W)-H(U\mid VW)
 =[H(UW)-H(W)]-[H(UVW)-H(VW)].
\]
We also use $H(U\mid W)\ge0$ and $I(U;V\mid W)\ge0$.

Relevant to secret sharing are the additional consequences:

\begin{lemma}\label{lem:dealer}
Let $A$ be unauthorised in a perfect scheme.
\begin{enumerate}
\item If $Aa$ is authorised, then
\[
 H(a\mid A)\ge H(S).
\]
\item If both $Aa$ and $Ab$ are authorised, where $a\ne b$, then
\[
 I(a;b\mid A)\ge H(S).
\]
\end{enumerate}
\end{lemma}

\begin{proof}
For the first assertion,
\[
 H(a\mid A)\ge I(S;a\mid A)
   =H(S\mid A)-H(S\mid Aa)=H(S).
\]
Here the first inequality follows from
$I(S;a\mid A)=H(a\mid A)-H(a\mid S,A)\le H(a\mid A)$.

For the second assertion, we expand both conditional mutual informations:
\begin{align*}
&I(a;b\mid A)-I(a;b\mid A,S)\\
={}&
 \bigl(H(Aa)+H(Ab)-H(A)-H(Aab)\bigr)\\
&\quad-
 \bigl(H(ASa)+H(ASb)-H(AS)-H(ASab)\bigr)\\
={}&
 \bigl(H(AS)-H(A)\bigr)
 -\bigl(H(ASa)-H(Aa)\bigr)\\
&\quad
 -\bigl(H(ASb)-H(Ab)\bigr)
 +\bigl(H(ASab)-H(Aab)\bigr)\\
={}&
 H(S\mid A)-H(S\mid Aa)-H(S\mid Ab)+H(S\mid Aab).
\end{align*}
Because \(A\) is unauthorised, perfect privacy gives
\(H(S\mid A)=H(S)\).  Because \(Aa\) and \(Ab\) are authorised,
correctness gives
\[
H(S\mid Aa)=H(S\mid Ab)=0.
\]
Moreover, \(Aab\) is authorised by monotonicity, so
\(H(S\mid Aab)=0\).  Therefore
\[
I(a;b\mid A)-I(a;b\mid A,S)=H(S).
\]
Finally, \(I(a;b\mid A,S)\geq0\), and hence $I(a;b\mid A)\geq H(S)$.
\end{proof}

\section{A matching-selector access structure}

Fix $k\ge3$, and take all indices in $\Zk$.  The participant set is
\[
 \cP_k=\{X_0,\ldots,X_{k-1},Y_0,\ldots,Y_{k-1},
                 h_0,\ldots,h_{k-1}\}.
\]
For each $s\in\Zk$, define the perfect matching
\[
 M_s=\bigl\{\{Y_i,X_{i+s}\}:i\in\Zk\bigr\}.
\]
A coalition $A$ is authorised in $\Gamma_k$ if and only if there is some
$s\in\Zk$ such that
\begin{equation}
 h_s\in A
 \quad\text{and}\quad
 A\cap\{Y_i,X_{i+s}\}\ne\varnothing
 \quad\text{for every }i\in\Zk.
\end{equation}
In formula notation, the access structure is given by:
\begin{equation}
 \bigvee_{s\in\Zk}\left(
 h_s\wedge\bigwedge_{i\in\Zk}(X_{i+s}\vee Y_i)
 \right).
\end{equation}

The key claim of this paper is the following, which we prove in Section~\ref{s:6} using machinery in Sections~\ref{s:4} and~\ref{s:5}.
\begin{theorem}\label{thm:charged}
Every perfect secret-sharing scheme for $\Gamma_k$ satisfies
\begin{equation}\label{eq:charged}
 \sum_{i=0}^{k-1}\bigl(H(X_i)+H(Y_i)\bigr)
 \ge k(k+1)H(S).
\end{equation}
\end{theorem}

\section{The entropy certificate}
\label{s:4}
Although access-structure indices are taken modulo $k$, the notation
$X_{a:b}$ below means the non-wrapping range
\[
 X_{a:b}=\{X_a,X_{a+1},\ldots,X_b\},
\]
which is empty when $a>b$.  We use the analogous notation for $Y$ and $h$.

Consider the following sum $\cB=\cH+\cC$ of $4k$ nonnegative quantities:
\begin{align}
\cH={}&H(X_0\mid X_{1:k-1}Y_1h_{0:k-2})
      +H(X_0\mid Y_{0:k-2}h_{0:k-1}) \notag\\
 &+I(X_0;h_0\mid X_{1:k-1})
      +I(X_0;h_0\mid Y_{1:k-1}) \notag\\
 &+\sum_{t=1}^{k-2}\!\left(
    I(X_0;h_0\mid X_{1:k-1}Y_0h_{1:t})
   +I(X_0;h_0\mid Y_{1:k-1}h_{1:t})\right) \notag\\
 &+I(X_0;h_0\mid Y_{1:k-1}h_{1:k-1}),                 \tag{H}\label{eq:H}\\[1mm]
\cC={}&\sum_{j=1}^{k-1}I(X_0;X_j\mid X_{j+1:k-1}) \notag\\
 &+\sum_{j=0}^{k-2}I(X_0;Y_0\mid X_{1:j})
   +I(X_0;Y_0\mid X_{1:k-1}h_0).                       \tag{C}\label{eq:C}
\end{align}

We apply Lemma~\ref{lem:dealer} to precisely the following $k+1$ terms;
all remaining terms are bounded only by nonnegativity:
\begin{itemize}
\item the two conditional entropies in the first line of~\eqref{eq:H};
\item the first term in the parenthesis for each $1\le t\le k-2$;
\item the last term of~\eqref{eq:C}.
\end{itemize}
We now check the authorisation premises required for those applications.

The conditioning coalition of the first conditional entropy contains every
$X$ except $X_0$, together with $Y_1$ and
$h_0,\ldots,h_{k-2}$.  For a present $h_s$, the edge of $M_s$ whose
$X$-endpoint is $X_0$ has other endpoint $Y_{-s}$.  This is $Y_1$ only for
the absent helper $h_{k-1}$.  The coalition is therefore unauthorised, while
adding $X_0$ authorises it.  The conditioning coalition of the second
conditional entropy contains no $X$ and misses $Y_{k-1}$, so it is
unauthorised.  After adding $X_0$, it covers $M_1$ and contains $h_1$.

For the first term in the parenthesis indexed by $t$, the conditioning
coalition is $X_{1:k-1}Y_0h_{1:t}$.  For every present $h_s$, the edge
$\{Y_{-s},X_0\}$ is uncovered because $s\ne0$.  Adding $X_0$ authorises the
coalition under $h_1$, and adding $h_0$ also authorises it.  Finally, $X_{1:k-1}h_0$ is unauthorised, while
adjoining either $X_0$ or $Y_0$ completes $M_0$.  Lemma~\ref{lem:dealer}
therefore gives
\begin{equation}\label{eq:seed}
 \cB\ge(k+1)H(S).
\end{equation}

\section{Symmetries and orbit notation}
\label{s:5}
For $\varepsilon\in\{+1,-1\}$ and $p,q\in\Zk$, define
\begin{align}
 D_{\varepsilon,p,q}:\quad
 &X_i\mapsto X_{\varepsilon i+p},\qquad
  Y_i\mapsto Y_{\varepsilon i+q},\qquad
  h_s\mapsto h_{\varepsilon s+p-q},                         \tag{D}\label{eq:D}\\
 W_{\varepsilon,p,q}:\quad
 &X_i\mapsto Y_{\varepsilon i+p},\qquad
  Y_i\mapsto X_{\varepsilon i+q},\qquad
  h_s\mapsto h_{-\varepsilon s-p+q}.                        \tag{W}\label{eq:W}
\end{align}

The inverse maps are
\[
 D_{\varepsilon,p,q}^{-1}
   =D_{\varepsilon,-\varepsilon p,-\varepsilon q},
 \qquad
 W_{\varepsilon,p,q}^{-1}
   =W_{\varepsilon,-\varepsilon q,-\varepsilon p}.
\]

We argue that these maps preserve authorisation.  Indeed,
\(D_{\varepsilon,p,q}\) sends \(h_s\) to
\(h_{\varepsilon s+p-q}\) and sends every edge
\[
 \{Y_i,X_{i+s}\}\in M_s
\]
to an edge of \(M_{\varepsilon s+p-q}\).  Similarly,
\(W_{\varepsilon,p,q}\) sends \(h_s\) to
\(h_{-\varepsilon s-p+q}\) and sends every edge of \(M_s\) to an edge of
\(M_{-\varepsilon s-p+q}\).  Consequently, if \(A\) contains \(h_s\) and
covers \(M_s\), then \(g(A)\) contains the corresponding helper and covers
the corresponding matching, so \(g(A)\) is authorised.  Applying the same
argument to \(g^{-1}\) gives the converse.  Hence
\[
 A\text{ is authorised}
 \quad\Longleftrightarrow\quad
 g(A)\text{ is authorised}
\]
for every displayed map \(g\); in particular, unauthorised coalitions are
also preserved.

For $k\ge3$, these $4k^2$ maps are distinct: whether a map preserves or
exchanges the $X$- and $Y$-classes determines whether it is a $D$ or a $W$;
the images of $X_0,Y_0$ determine $p,q$; and the image of $X_1$ determines
$\varepsilon$, since $1\not\equiv-1\pmod k$.  They form a group $\cG$, as
is also seen from
\begin{align*}
 D_{a,p,q}D_{b,r,t}&=D_{ab,ar+p,at+q},&
 D_{a,p,q}W_{b,r,t}&=W_{ab,ar+q,at+p},\\
 W_{a,p,q}D_{b,r,t}&=W_{ab,ar+p,at+q},&
 W_{a,p,q}W_{b,r,t}&=D_{ab,ar+q,at+p}.
\end{align*}

If \(F\) is a linear expression in joint entropies, let \(gF\) denote the
expression obtained by replacing every participant and coalition appearing
in \(F\) by its image under \(g\).  Define the group average of \(F\) by
\[
 \operatorname{Av}(F)
   :=\frac1{|\mathcal G|}\sum_{g\in\mathcal G}gF.
\]
In particular, for a coalition \(A\) we use the
abbreviation
\[
 [A]
   :=\operatorname{Av}(H(A))
   =\frac1{|\mathcal G|}\sum_{g\in\mathcal G}H(gA).
\]

Multiplication by any fixed \(g_0\in\mathcal G\) merely permutes the terms
in the group average.  Consequently,
\[
 \operatorname{Av}(g_0F)=\operatorname{Av}(F)
 \quad\text{and, in particular,}\quad
 [g_0A]=[A].
\]

Finally, the proof of~\eqref{eq:seed} applies to $g\cB$: the map $g$ preserves every
authorisation and nonauthorisation premise used above, while all other terms
remain nonnegative. 

Therefore, for every \(g\in\mathcal G\),
\begin{equation}\label{eq:all-images}
 g\mathcal B\geq(k+1)H(S),
 \qquad\text{and hence}\qquad
 \operatorname{Av}(\mathcal B)\geq(k+1)H(S).
\end{equation}

\section{Tracking cancellations of the averaged certificate}
\label{s:6}
Set
\[
 P=X_{0:k-2},\qquad \mathbf X=X_{0:k-1},\qquad
 X^-=X_{1:k-1},\qquad Y^-=Y_{1:k-1},
\]
\[
 R=h_{0:k-2},\qquad \mathbf h=h_{0:k-1}.
\]
Split~\eqref{eq:H} into its first two, next two, summed, and final parts,
called $H_1,H_2,H_3,H_4$.  Call the three respective parts of
\eqref{eq:C} $C_1,C_2,C_3$.

We derive the averaged expansion of each part.  Whenever we use
$[A]=[B]$, we justify the equality by giving a mapping $g$ such that $B=gA$.

\paragraph{The term $C_1$.}
The chain rule gives
\[
 \sum_{j=1}^{k-1}I(X_0;X_j\mid X_{j+1:k-1})
 =I(X_0;X^-).
\]
Expanding with~\eqref{eq:cmi-four}, and using
$[X^-]=[P]$ via $D_{1,-1,-1}$, gives
\begin{equation}\label{eq:C1}
 \Av(C_1)=[X_0]+[P]-[\varnothing]-[\mathbf X].
\end{equation}

\paragraph{The term $C_2$.}

Put \(A_j=X_{1:j}\), so that \(A_0=\varnothing\) and
\(A_{k-1}=X^-\).  Replacing each conditional mutual information by
\[
 I(U;V\mid W)=H(WU)+H(WV)-H(W)-H(WUV)
\]
gives
\begin{align*}
 \Av(C_2)
  =\sum_{j=0}^{k-2}
   \bigl(
     [A_jX_0]+[A_jY_0]-[A_j]-[A_jX_0Y_0]
   \bigr).
\end{align*}
For every \(0\leq j\leq k-2\), the following symmetry relations hold:
\[
\begin{aligned}
 [A_jX_0]&=[A_{j+1}]
     &&\text{via }D_{1,1,1},\\
 [A_jX_0Y_0]&=[A_{j+1}Y_0]
     &&\text{via }D_{1,1,0}.
\end{aligned}
\]
Substituting these relations makes the two resulting sums telescope:
\begin{align*}
 \Av(C_2)
 &=\sum_{j=0}^{k-2}
      \bigl([A_{j+1}]-[A_j]\bigr)
   +\sum_{j=0}^{k-2}
      \bigl([A_jY_0]-[A_{j+1}Y_0]\bigr)\\
 &=[A_{k-1}]-[A_0]+[A_0Y_0]-[A_{k-1}Y_0]\\
 &=[X^-]-[\varnothing]+[Y_0]-[X^-Y_0]\\
 &=[P]+[X_0]-[PY_0].
\end{align*}
The last equality uses
\[
\begin{aligned}
 [X^-]&=[P]       &&\text{via }D_{1,-1,-1},\\
 [Y_0]&=[X_0]     &&\text{via }W_{1,0,0},\\
 [X^-Y_0]&=[PY_0] &&\text{via }D_{1,-1,0},
\end{aligned}
\qquad\text{and}\qquad
[\varnothing]=0.
\]
Therefore
\begin{equation}\label{eq:C2}
 \Av(C_2)=[X_0]+[P]-[PY_0].
\end{equation}

\paragraph{The term $C_3$.}
Direct expansion gives four terms.  Using
$[X^-h_0]=[Ph_0]$ via $D_{1,-1,-1}$,
\begin{equation}\label{eq:C3}
 \Av(C_3)=-[Ph_0]+[\mathbf Xh_0]+[X^-Y_0h_0]
           -[\mathbf XY_0h_0].
\end{equation}

\paragraph{The term \(H_1\).}
The expression \(H_1\) is the sum of the first two conditional entropies
in~\eqref{eq:H}:
\[
 H_1
 =H(X_0\mid X^-Y_1R)
  +H(X_0\mid Y_{0:k-2}\mathbf h).
\]
Replacing each conditional entropy by
\[
 H(U\mid W)=H(UW)-H(W)
\]
gives
\begin{align*}
 \Av(H_1)
  ={}&[\mathbf XY_1R]-[X^-Y_1R]\\
    &+[X_0Y_{0:k-2}\mathbf h]
      -[Y_{0:k-2}\mathbf h].
\end{align*}
The required symmetry relations are
\[
\begin{aligned}
 [\mathbf XY_1R]&=[\mathbf XY_0R]
     &&\text{via }D_{1,-1,-1},\\
 [X^-Y_1R]&=[PY_0R]
     &&\text{via }D_{1,-1,-1},\\
 [X_0Y_{0:k-2}\mathbf h]&=[PY_0\mathbf h]
     &&\text{via }W_{1,0,0},\\
 [Y_{0:k-2}\mathbf h]&=[P\mathbf h]
     &&\text{via }W_{1,0,0}.
\end{aligned}
\]
Substituting these relations yields
\begin{equation}\label{eq:H1}
 \Av(H_1)
 =-[PY_0R]+[\mathbf XY_0R]
  -[P\mathbf h]+[PY_0\mathbf h].
\end{equation}

\paragraph{The term $H_1$.}
The two conditional entropies expand into four terms.  The required
orbit identities are
\begin{align*}
 [\mathbf XY_1R]&=[\mathbf XY_0R] &&\text{via }D_{1,-1,-1},\\
 [X^-Y_1R]&=[PY_0R] &&\text{via }D_{1,-1,-1},\\
 [X_0Y_{0:k-2}\mathbf h]&=[PY_0\mathbf h]
     &&\text{via }W_{1,0,0},\\
 [Y_{0:k-2}\mathbf h]&=[P\mathbf h]
     &&\text{via }W_{1,0,0}.
\end{align*}
Therefore
\begin{equation}\label{eq:H1}
 \Av(H_1)=-[PY_0R]+[\mathbf XY_0R]-[P\mathbf h]+[PY_0\mathbf h].
\end{equation}

\paragraph{The term \(H_2\).}
The expression \(H_2\) is the sum of the next two conditional mutual
informations in~\eqref{eq:H}:
\[
 H_2
 =I(X_0;h_0\mid X^-)+I(X_0;h_0\mid Y^-).
\]
Replacing each conditional mutual information by
\[
 I(U;V\mid W)
 =H(WU)+H(WV)-H(W)-H(WUV)
\]
gives
\begin{align*}
 \Av(H_2)
  ={}&[\mathbf X]+[X^-h_0]-[X^-]-[\mathbf Xh_0]\\
    &+[X_0Y^-]+[Y^-h_0]-[Y^-]-[X_0Y^-h_0].
\end{align*}
The required symmetry relations are
\[
\begin{aligned}
 [X^-]&=[P]
     &&\text{via }D_{1,-1,-1},\\
 [X^-h_0]&=[Ph_0]
     &&\text{via }D_{1,-1,-1},\\
 [Y^-]&=[P]
     &&\text{via }W_{-1,-1,-1},\\
 [Y^-h_0]&=[Ph_0]
     &&\text{via }W_{-1,-1,-1},\\
 [X_0Y^-]&=[PY_0]
     &&\text{via }W_{1,0,-1},\\
 [X_0Y^-h_0]&=[X^-Y_0h_0]
     &&\text{via }W_{1,0,0}.
\end{aligned}
\]
Substituting these relations yields
\begin{align*}
 \Av(H_2)
  ={}&[\mathbf X]+[Ph_0]-[P]-[\mathbf Xh_0]\\
    &+[PY_0]+[Ph_0]-[P]-[X^-Y_0h_0].
\end{align*}
Therefore
\begin{equation}\label{eq:H2}
 \Av(H_2)
 =-2[P]+[\mathbf X]+[PY_0]+2[Ph_0]
  -[\mathbf Xh_0]-[X^-Y_0h_0].
\end{equation}

\paragraph{The term \(H_3\).}
For \(1\leq t\leq k-2\), put \(K_t=h_{1:t}\), and let \(H_{3,t}\)
denote the two summands of \(H_3\) indexed by \(t\):
\[
 H_{3,t}
 =I(X_0;h_0\mid X^-Y_0K_t)
  +I(X_0;h_0\mid Y^-K_t).
\]
Expanding both conditional mutual informations gives
\begin{align*}
 \Av(H_{3,t})
 ={}&[\mathbf XY_0K_t]+[X^-Y_0h_0K_t]
       -[X^-Y_0K_t]-[\mathbf XY_0h_0K_t]\\
   &+[X_0Y^-K_t]+[Y^-h_0K_t]
       -[Y^-K_t]-[X_0Y^-h_0K_t].
\end{align*}
Two pairs cancel because
\[
\begin{aligned}
 [X^-Y_0h_0K_t]&=[X_0Y^-h_0K_t]
     &&\text{via }W_{-1,0,0},\\
 [X^-Y_0K_t]&=[X_0Y^-K_t]
     &&\text{via }W_{-1,0,0}.
\end{aligned}
\]
The four terms that remain satisfy
\[
\begin{aligned}
 [\mathbf XY_0K_t]
   &=[\mathbf XY_0h_{0:t-1}]
     &&\text{via }D_{1,-1,0},\\
 [\mathbf XY_0h_0K_t]
   &=[\mathbf XY_0h_{0:t}]
     &&\text{by literal equality},\\
 [Y^-h_0K_t]
   &=[Ph_{0:t}]
     &&\text{via }W_{-1,-1,-1},\\
 [Y^-K_t]
   &=[Ph_{0:t-1}]
     &&\text{via }W_{-1,0,-1}.
\end{aligned}
\]
Consequently,
\begin{align*}
 \Av(H_{3,t})
  ={}&[\mathbf XY_0h_{0:t-1}]
      -[\mathbf XY_0h_{0:t}]\\
    &+[Ph_{0:t}]-[Ph_{0:t-1}].
\end{align*}
Define
\[
 E_t=[Ph_{0:t}]-[\mathbf XY_0h_{0:t}].
\]
The preceding display is exactly
\[
 \Av(H_{3,t})=E_t-E_{t-1}.
\]
Since \(H_3=\sum_{t=1}^{k-2}H_{3,t}\), these differences telescope:
\begin{align*}
 \Av(H_3)
 &=\sum_{t=1}^{k-2}(E_t-E_{t-1})\\
 &=E_{k-2}-E_0\\
 &=\bigl([PR]-[\mathbf XY_0R]\bigr)
   -\bigl([Ph_0]-[\mathbf XY_0h_0]\bigr).
\end{align*}
Therefore
\begin{equation}\label{eq:H3}
 \Av(H_3)
 =-[Ph_0]+[\mathbf XY_0h_0]+[PR]-[\mathbf XY_0R].
\end{equation}

\paragraph{The term \(H_4\).}
The final part of \(\mathcal H\) is
\[
 H_4=I(X_0;h_0\mid Y^-h_{1:k-1}).
\]
Expanding this conditional mutual information gives
\begin{align*}
 \Av(H_4)
  ={}&[X_0Y^-h_{1:k-1}]+[Y^-\mathbf h]\\
    &-[Y^-h_{1:k-1}]-[X_0Y^-\mathbf h].
\end{align*}
The required symmetry relations are
\[
\begin{aligned}
 [X_0Y^-h_{1:k-1}]&=[PY_0R]
     &&\text{via }W_{1,0,-1},\\
 [Y^-\mathbf h]&=[P\mathbf h]
     &&\text{via }W_{1,0,-1},\\
 [Y^-h_{1:k-1}]&=[PR]
     &&\text{via }W_{1,0,-1},\\
 [X_0Y^-\mathbf h]&=[PY_0\mathbf h]
     &&\text{via }W_{1,0,-1}.
\end{aligned}
\]
Substituting these relations yields
\begin{equation}\label{eq:H4}
 \Av(H_4)
 =-[PR]+[PY_0R]+[P\mathbf h]-[PY_0\mathbf h].
\end{equation}

For convenience, the seven averaged expansions are collected below:
\begin{equation}\label{eq:ledger}
\begin{aligned}
 \Av(C_1)
   &=[X_0]+[P]-[\mathbf X],\\
 \Av(C_2)
   &=[X_0]+[P]-[PY_0],\\
 \Av(C_3)
   &=-[Ph_0]+[\mathbf Xh_0]+[X^-Y_0h_0]
      -[\mathbf XY_0h_0],\\[1mm]
 \Av(H_1)
   &=-[PY_0R]+[\mathbf XY_0R]-[P\mathbf h]
      +[PY_0\mathbf h],\\
 \Av(H_2)
   &=-2[P]+[\mathbf X]+[PY_0]+2[Ph_0]
      -[\mathbf Xh_0]-[X^-Y_0h_0],\\
 \Av(H_3)
   &=-[Ph_0]+[\mathbf XY_0h_0]+[PR]
      -[\mathbf XY_0R],\\
 \Av(H_4)
   &=-[PR]+[PY_0R]+[P\mathbf h]
      -[PY_0\mathbf h].
\end{aligned}
\end{equation}

Every entry except the two copies of $[X_0]$ cancels, so
\begin{equation}\label{eq:average-identity}
 \Av(\cB)=2[X_0].
\end{equation}

It remains only to count how often the singleton entropies occur in
\([X_0]\).  Under the maps \(D_{\varepsilon,p,q}\), \(X_0\) is sent to
\(X_p\).  For each fixed \(X_p\), there are two choices of
\(\varepsilon\) and \(k\) choices of \(q\), so \(H(X_p)\) occurs \(2k\)
times.  Under the maps \(W_{\varepsilon,p,q}\), \(X_0\) is sent to \(Y_p\),
and similarly each \(H(Y_p)\) occurs \(2k\) times. Therefore
\begin{equation}\label{eq:singleton-average}
 [X_0]=\frac1{2k}\sum_{i=0}^{k-1}\bigl(H(X_i)+H(Y_i)\bigr).
\end{equation}
Combining~\eqref{eq:all-images}, \eqref{eq:average-identity}, and
\eqref{eq:singleton-average} gives
\[
 \frac1k\sum_{i=0}^{k-1}\bigl(H(X_i)+H(Y_i)\bigr)
 \ge(k+1)H(S),
\]
which proves Theorem~\ref{thm:charged}.

\section{Complexity and bit-length consequences}

We make the complexity measures explicit.  For a perfect scheme
$\Pi=(S,(Z_a)_{a\in\cP})$, define its normalised total and maximum share
entropies by
\[
 \tau(\Pi)=\frac{\sum_{a\in\cP}H(Z_a)}{H(S)},
 \qquad
 \sigma(\Pi)=\frac{\max_{a\in\cP}H(Z_a)}{H(S)}.
\]
For an access structure $\Gamma$, let
\[
 \tau(\Gamma)=\inf_{\Pi\text{ perfect for }\Gamma}\tau(\Pi),
 \qquad
 \sigma(\Gamma)=\inf_{\Pi\text{ perfect for }\Gamma}\sigma(\Pi).
\]

\begin{corollary}\label{cor:complexity}
For $\Gamma_k$, where $n=3k$,
\[
 \tau(\Gamma_k)\ge k(k+2)=\frac{n^2}{9}+\frac{2n}{3},
 \qquad
 \sigma(\Gamma_k)\ge\frac{k+1}{2}=\frac n6+\frac12.
\]
In particular, $\tau(\Gamma_k)=\Omega(n^2)$ and
$\sigma(\Gamma_k)=\Omega(n)$.
\end{corollary}

\begin{proof}
Theorem~\ref{thm:charged} gives the entropy sum of $2k$
shares and by dividing the sum by the number $2k$ of terms yields that at least one share must have entropy at least $\frac{k+1}{2}H(S)$, which yields the maximum-share bound.  

For $\tau$, we note the
coalition $\mathbf X$ is unauthorised because it contains no helper, while
$\mathbf Xh_s$ is authorised for every $s$.  Lemma~\ref{lem:dealer} then gives
\[
 H(h_s\mid\mathbf X)\ge H(S),
\]
and hence $H(h_s)\ge H(S)$, since conditioning cannot increase entropy.  Adding the \(k\) helper contributions to the charged-share bound finishes the result:
\[
 \sum_{a\in\mathcal P_k}H(a)
 =\sum_{i=0}^{k-1}\bigl(H(X_i)+H(Y_i)\bigr)
   +\sum_{s=0}^{k-1}H(h_s)
 \geq\bigl(k(k+1)+k\bigr)H(S)
 =k(k+2)H(S).
\]
\end{proof}

\paragraph{Alphabet-size bounds.} Suppose the secret is uniform on
a finite alphabet $\mathcal S$, and participant $a$ has a share in a finite
alphabet $\mathcal A_a$.  Then
\[
 H(S)=\log_2|\mathcal S|,
 \qquad H(Z_a)\le\log_2|\mathcal A_a|.
\]
Consequently,
\begin{equation}\label{eq:alphabet-total}
 \frac{\sum_{a\in\cP_k}\log_2|\mathcal A_a|}
      {\log_2|\mathcal S|}
 \ge k(k+2),
\end{equation}
and
\begin{equation}\label{eq:alphabet-max}
 \frac{\max_{a\in\cP_k}\log_2|\mathcal A_a|}
      {\log_2|\mathcal S|}
 \ge\frac{k+1}{2}.
\end{equation}
If the secret is a uniform $m$-bit string, then
$\log_2|\mathcal S|=m$.  If shares are represented by fixed-length binary
strings, participant $a$ requires at least
$\ell_a=\lceil\log_2|\mathcal A_a|\rceil$ bits.  Since
$\ell_a\ge\log_2|\mathcal A_a|$, the exact bit-length analogues are
\[
 \sum_{a\in\cP_k}\ell_a\ge k(k+2)m,
 \qquad
 \max_{a\in\cP_k}\ell_a\ge\frac{k+1}{2}m.
\]

\section*{Acknowledgments}

The primary arguement was constructed with ChatGPT 5.6 Sol, with additional computational exploration from OpenAI's Codex.  The mathematical proof
presented here was independently checked line by line by the author, who takes full responsibility for it.

\bibliographystyle{alpha}
\bibliography{cireferences}

\end{document}